\documentclass[11pt, onecolumn, technote]{IEEEtran}

\usepackage{amsthm,amsmath,amssymb,graphicx}
\RequirePackage[colorlinks,citecolor=blue,urlcolor=blue]{hyperref}

\theoremstyle{plain}
\newtheorem{theorem}{Theorem}
\newtheorem{lem}[theorem]{Lemma}
\theoremstyle{remark}
\newtheorem{rem}[theorem]{Remark}
\theoremstyle{corollary}
\newtheorem{cor}[theorem]{Corollary}

\begin{document}

\title{Mutual Information and Conditional Mean Prediction Error}

\author{
\IEEEauthorblockN{Clive G. Bowsher$^1$} 
and 
\IEEEauthorblockN{Margaritis Voliotis} \\ 
School of Mathematics, \\ University of Bristol, U.K. 
\thanks{$^1$ corresponding author: c.bowsher@bristol.ac.uk} 
}

\maketitle

\begin{abstract}
Mutual information is fundamentally important for measuring statistical
dependence between variables and for quantifying information transfer by signaling and communication mechanisms. It can, however, be challenging to evaluate
for physical models of such mechanisms and to estimate reliably from data. Furthermore,
its relationship to better known statistical procedures is still poorly
understood. Here we explore new connections between mutual information
and regression-based dependence measures, $\nu^{-1}$, that utilise the determinant of the second-moment matrix of the conditional mean prediction error. We examine
convergence properties as $\nu\rightarrow0$ and establish sharp lower
bounds on mutual information and capacity of the form $\mathrm{log}(\nu^{-1/2})$. The bounds are tighter than 
lower bounds based on the Pearson correlation and ones derived using average mean square-error rate distortion arguments. Furthermore, their estimation is feasible using techniques from nonparametric regression.  
As an illustration we provide bootstrap confidence intervals for the
lower bounds which, through use of a composite estimator, substantially improve upon inference about mutual information based on $k$-nearest neighbour estimators alone.
\end{abstract}


\section{Introduction}

Mutual information is fundamentally important for measuring statistical
dependence between variables \cite{Linfoot1957,Joe1989,Brillinger2007,Reshef2011},
and for quantifying information transfer by engineered or naturally
occurring communication systems \cite{Shannon1948,Brennan2012}.
Statistical analysis using mutual information has been particularly
influential in neuroscience \cite{BialekSpikes1999}, and is becoming
so in systems biology for studying the biomolecular signaling networks
used by cells to detect, process and act upon the chemical signals
they receive \cite{Cheong2011,Bowsher2012,Bowsher2013}.
It can, however, be challenging to estimate mutual information reliably
with available sample sizes \cite{Panzeri2012}, and difficult to
derive mutual information and capacity exactly using mechanistic
models of the `channels' via which signals are conveyed. Furthermore, connections between mutual information and better
known statistical procedures such as regression, and their associated dependence measures, are still poorly understood. In order to address these challenges, the relationship between mutual information and the error incurred by estimation (or `prediction') using the conditional mean is now receiving attention. The focus has been on minimum mean square estimation error or, more generally, its average across the elements of the vector being estimated \cite{Prelov2004,Guo2005,Guo2008}. Instead, we focus on  
connections between mutual information and regression-based dependence measures, $\nu^{-1}$, that utilise the determinant of the second-moment matrix of the conditional mean prediction error. We examine convergence
properties as $\nu\rightarrow0$, and establish sharp lower bounds on
mutual information of the form $\mathrm{log}(\nu^{-1/2})$. The bounds are tighter than 
lower bounds based on the Pearson correlation and ones derived using average mean square-error rate distortion arguments.

The \emph{mutual information} between 2 random vectors $X$ and $Z$, written
$I(X;Z),$ is the Kullback-Leibler divergence between their joint
distribution and the product of their marginal distributions \cite{CoverThomasBook2006}. Mutual information thus measures the divergence between the joint distribution of $(X,Z)$ and the distribution in which $X$ and $Z$ are independent but have the same marginals. $I(X;Z)$ has desirable properties
as a measure of statistical dependence: it satisfies (after monotonic
transformation) all 7 of R\'{e}nyi's postulates \cite{Renyi1959} for
a dependence measure between 2 random variables, and underlies the
recently introduced maximal information coefficient \cite{Reshef2011}
for detecting pairwise associations in large data sets. Importantly, mutual information captures nonlinear dependence and dependence arising from moments higher than the conditional mean.

A decision-theoretic interpretation is indicative of the broad applicability of mutual information as a summary measure of statistical dependence. It can be shown \cite{Bernardo1979} that $I(X;Z)$ is equal to the increase in expected utility from reporting the posterior distribution of either of the two random vectors based on observation of the other, compared to reporting its marginal distribution---for example, reporting $p(Z|X=x)$ instead of $p(Z)$. This holds when the utility
function is a smooth, proper, local score function---as appropriate
for scientific reporting of distributions as `pure inferences' \cite{BernardoSmithBayesianTheoryBook2000}---because the logarithmic score function is the only score function having all these properties. 
In information theory, the supremum of $I(X;Z)$ over the set of allowed
input distributions $\mathcal{F}$, termed the \emph{information capacity}, equates to the maximal
errorless rate of information transmission over a noisy
channel when the channel is used for long times \cite{Shannon1948}.

The above discussion makes clear that, from a rich variety of perspectives, mutual information is fundamentally important for measuring statistical
dependence and for quantifying information
transfer by signaling and communication mechanisms.

Our general setting may be depicted as
\begin{eqnarray}
X \rightarrow Y \rightarrow Z,
\label{picture}
\end{eqnarray}
with $(X,Y,Z)$ a real-valued random vector. Here $Z$ is conditionally independent of $X$ given $Y$. The conditional distribution of $Z$ given $X$ is determined by some physical mechanism, whose `internal' variables are denoted by $Y$ in Eq. \ref{picture}. Such a mechanism is often termed a channel in information theory, although we do not restrict attention to signaling and communication channels here. Often we have in mind situations where $X$ causes $Z$ but not \emph{vice versa}, and the conditional distribution of $Z$ given $X$ does not depend on the experimental `regime' giving rise to the distribution of $X$ \cite{Dawid2010}. There is always an asymmetry between $X$ and $Z$ in the general setting we consider. We term $X$ the \emph{input} or \emph{treatment} because its marginal distribution can in principle be any distribution (although we may wish to restrict attention to particular classes thereof).  In contrast, the \emph{output} or \emph{response} $Z$ is the realisation of the mechanism given the input $X$. In general, not all marginal distributions for $Z$ can be obtained for a given mechanism by appropriate choice of the marginal of $X$. When analysing the probabilistic properties of physical models of mechanisms, the distribution (or set of distributions) for the input $X$ is given, but the marginal distribution of $Z$ is often unknown. In experimental settings, the input distribution is taken not to affect the conditional distribution of $Z$ given $X$, or else can sometimes be directly manipulated. 

Examples of our general setting include experimental design with $X$ as the treatment and $Z$ the response of interest; and signaling or communication channels with $X$ as the input signal and $Z$ its noisy representation. An example of a scientific area of application is the current effort to understand the biomolecular signaling mechanisms
used by living cells to relay the chemical signals
they receive from their environment \cite{Cheong2011}. Here, the interest is both in understanding why some biomolecular mechanisms perform better than others, and in measuring experimentally in the laboratory the mutual information between $X$ and $Z$ or the information capacity. Broadly speaking, the first involves deriving dependence measures between $X$ and $Z$ for different stochastic mechanisms (given a particular input distribution). The second might involve, for example, nonparametric estimation of the mutual information between the concentration of a chemical treatment applied to the cells and the level of an intracellular, biochemical output.

We note, however, that the formal statements of our results do not require any particular interpretation of $X$ and $Z$. Rather, the general setting just described motivates the results and places them in context. A sequential reading of Equations 2 to 12 inclusive provides a convenient preview of our theoretical results establishing lower bounds on mutual information and information capacity.

\section{Setup and Notation}

For random vectors $X$ and $Z$,
we define 
$\nu(Z|X)=\mathbb{\mathrm{\det}\left(E\mathit{\left\{ \mathbb{V}\left[Z|X\right]\right\} }\right)}\mathrm{/\det}\left(\mathbb{V}\left[\mathit{Z}\right]\right)$, where $\mathbb{V}$ denotes a covariance matrix.
In general, $\nu(Z|X)$ is not equal
to $\nu(X|Z)$. For 2 scalar random variables, $\nu(Z|X)$ is equal
to the minimum mean square estimation error or minimum MSEE for estimation
of $Z$ using $X$, normalised by the variance of $Z$ (since $\mathbb{E}\left[Z|X\right]$
is the optimal estimator). We denote the optimal estimation
or `prediction' error by $e(Z|X)=Z-\mathbb{E}\left[Z|X\right]$. In general, $\nu(Z|X)$ is the ratio of the determinant of the second-moment matrix of the
error $e(Z|X)$ and the determinant of the variance matrix of $Z$, that is
\begin{eqnarray}
\nu(Z|X)=\frac{ \mathbb{\mathrm{\det}\left(E\mathit{\left[e(Z|X)e(Z|X)^{\mathrm{T}}\right]}\right)} } {\mathrm{\det}\left(\mathbb{V}\left[\mathit{Z}\right]\right)}.
\end{eqnarray}

We will show that $\nu(Z|X)^{-1}$
provides a generalised measure of `signal-to-noise', applicable to
non-Gaussian settings, that relies on first and second conditional
moments of $Z$ given $X$ (via the law of total variance) rather than on all features of the joint distribution. We make few assumptions about the conditional density describing the mechanism or channel $f(z|x)$, except that the conditional mean $m(x)=\mathbb{E}\left[Z|X=x\right]$ is an invertible, continuously differentiable function of $x$. 
A central result of the paper (see Theorem \ref{thm:ZgXlwrbnd} and Corollary \ref{rdarg}) is then that
\begin{equation}
I(X;Z)\geq\mathrm{\log}\left\{\nu(Z|X)^{-1/2}\right\}
\geq \frac{-d_Z}{2}\mathrm{\log}\left\{ \frac{d_{Z}^{-1}\mathrm{tr}(\mathbb{E}\{\mathbb{V}[Z|X]\})}{\mathrm{[\det}(\mathbb{V}[Z])]^{d_{Z}^{-1}}}\right\} ,
\label{main}
\end{equation}
where all terms are evaluated under the joint density for $(X,Z)$ implied by the channel $f(z|x)$ and a Gaussian density for the transformed input, $m(X)$. Here $d_Z$ is the dimension of the vector $Z$. The second term in Eq. \ref{main} is our lower bound utilising the determinant of the second-moment matrix of the prediction error of the conditional mean $\mathbb{E}[Z|X]$, while the third term instead utilises the average mean square error of that conditional mean. We discuss the relation of the third term to rate distortion arguments later in the paper. Notice that characterising the first and second conditional moments of the mechanism, $\mathbb{E}[Z|X]$ and $\mathbb{V}[Z|X]$, is enough (via the law of total variance) to evaluate the lower bound $\mathrm{\log}\left\{\nu(Z|X)^{-1/2}\right\}$ for a given Gaussian distribution of $m(X)$. Maximising the bound over such distributions then also yields a useful bound on the information capacity.

As a first step in analysing connections between mutual information and our regression-based measures,
we explore the relationship between the convergence to zero of $\nu(Z|X)$
or $\nu(X|Z)$, and the convergence of mutual information. For simplicity,
we analyse the bivariate case where the variable $X$ has finite support, for
example a finite collection of treatment concentrations in a cell
signaling experiment. We write $I$ for mutual information,
$H$ for discrete entropy and $h$ for differential entropy.

\section{Convergence properties}

\begin{theorem}
\label{thm:ZgXasymp}Let $(X_{n},Z_{n})$ be a sequence of pairs of real-valued
random variables, with the support of $X_{n}$ given by a finite set $\mathcal{X}_{n}$
($|\mathcal{X}_{n}|\geq2$ and bounded above by a constant $\forall n$).
Write $m_{n}(X_{n})$ for the function $\mathbb{E}[\breve{Z}_{n}|X_{n}]$,
where $\breve{Z}_{n}\triangleq Z_{n}\mathbb{V}[Z_{n}]^{-\frac{1}{2}}.$
Denote its support by $\mathcal{M}_{n} = \{m_{n}(x);x\in\mathcal{X}_{n}\}\subset\mathbb{R}$.
Let $\epsilon_{n}^{\ast}$ be 1/2 of the minimum distance between
any two points in $\mathcal{M}_{n}$ and $\epsilon^{\ast}\triangleq\inf\{\epsilon_{n}^{\ast}\}.$
Suppose that:

(1) $ \nu(Z_{n}|X_{n})\rightarrow0$ as $n\rightarrow\infty$; (2) the functions $m_{n}$ are one-to-one mappings from $\mathcal{X}_{n}$
to the real line; and (3) $\epsilon^{\ast}>0$, so that any pair of
points in a support $\mathcal{M}_{n}$ are separated by at least $2\epsilon^{\ast}$.
Then as $n\rightarrow\infty$ , $H(X_{n}|Z_{n})\rightarrow0$ and
therefore $H(X_{n})-I(X_{n};Z_{n})\rightarrow0$. 
\end{theorem}
In Theorem \ref{thm:ZgXasymp}, the response variable $Z$ is real-valued
and can, for example, be either a continuous or discrete random variable.
Theorem \ref{thm:ZgXasymp} establishes the convergence of the mutual
information $I(X;Z)$ to the entropy of $X$ as $\nu(Z|X)\rightarrow0$, under the condition that the conditional mean $\mathbb{E}[Z|X]$ is
an invertible function of $X$. By definition, $I(X;Z)$ cannot be greater than the entropy of $X$.
The intuition for the result in Theorem \ref{thm:ZgXasymp} is that the convergence of $\nu(Z|X)$ to zero enables construction
of a point estimator of $X$ (utilising the conditional expectation
$\mathbb{E}[Z|X]$) whose performance becomes `perfect' in this limit. The condition requiring invertibility of the conditional mean function would be needed even in the case where $Z$ is a deterministic function of $X$, otherwise $I(X;Z) \leq H(Z) < H(X)$.

We have rescaled the output $Z_{n}$
so that $\mathbb{V}[\breve{Z}_{n}]$ is constant at $1$ for all $n$.
In particular, Theorem \ref{thm:ZgXasymp} does \textit{not} require
the minimum mean square estimation error, $\mathbb{E}\{(\mathit{Z-\mathbb{E}[Z|X]})^{2}\},$
to converge to zero as $n\rightarrow\infty$. A physical example where the minimum MSEE diverges but Theorem \ref{thm:ZgXasymp} applies is given by a molecular signaling system, with $Z$ the number of output molecules, which is operating in the macroscopic (or large system size) limit where the dynamics of chemical concentrations are deterministic, conditional on the input $X$.
\begin{proof}
We have that $\nu(Z_{n}|X_{n})=\nu(\breve{Z}_{n}|X_{n})\rightarrow0$.
Since $\nu(\breve{Z}_{n}|X_{n})=\mathbb{E}\{\mathbb{V}[\breve{Z}_{n}|X_{n}]\}$
$=\mathbb{E}\{(\mathit{\breve{Z}_{n}-\mathbb{E}[\breve{Z}_{n}|X_{n}]})^{2}\}$,
it follows that $\breve{Z}_{n}-\mathbb{E}[\breve{Z}_{n}|X_{n}]$ converges
to zero in mean square (in $L^{2}$) and therefore $\breve{Z}_{n}-\mathbb{E}[\breve{Z}_{n}|X_{n}]\rightarrow^{pr}0$
(where $\rightarrow^{pr}$ denotes convergence in probability). Since
$H(X_{n}|\breve{Z}_{n})=H(X_{n}|Z_{n})$ $=H(X_{n})-I(X_{n},Z_{n}),$
we must establish that $H(X_{n}|\breve{Z}_{n})\rightarrow0$ as $n\rightarrow\infty$.
Consider estimating $X_{n}$ based on observation of $\breve{Z}_{n}$
using the estimator, $\widehat{X}_{n}(\breve{Z}_{n})$, which is equal
to the a point $x\in\mathcal{X}_{n}$ that minimises the distance
$|\breve{Z}_{n}-m_{n}(x)|$. Condition (3) applies not to $\mathbb{E}[Z_{n}|X_{n}=x]$
but to $\mathbb{E}[\breve{Z}_{n}|X_{n}=x].$ Let $z\in\mathbb{R},m_{n}\mathcal{\in M}_{n}$
and notice that if $|\mathit{z-}m_{n}|<\epsilon^{*},$ then $\mathit{|z-}m_{n}|<\epsilon_{n}^{*}<\mathit{|z-}m_{n}'|$,
that is $m_{n}$ is closer to $z$ than is any other point $m_{n}'$
in $\mathcal{M}_{n}$. Therefore, if $|\mathbb{\mathit{\breve{Z}_{n}-}E}[\breve{Z}_{n}|X_{n}]|<\epsilon^{*},$
then $\mathbb{E}[\breve{Z}_{n}|X_{n}]$ is located at the unique point
$m_{n}\mathcal{\in M}_{n}$ that minimises $|\breve{Z}_{n}-m_{n}|$
and, using condition (2), $\widehat{X}_{n}(\breve{Z}_{n})=X_{n}$,
that is, the estimator recovers $X_{n}$ without error$.$ The probability 
of estimation error, $p_{\mathrm{error}},$ satisfies 
\[
\mathbb{\mathit{p\mathrm{_{error}}=}P\left\{ \mathit{\hat{X}_{n}\left(\breve{Z}_{n}\right)\neq X_{n}}\right\} }\leq\mathbb{P}\left\{ |\mathbb{\mathit{\breve{Z}_{n}-}E}[\breve{Z}_{n}|X_{n}]|\geq\epsilon^{*}\right\} ,
\]
and therefore $p_{\mathrm{error}}\rightarrow0$ as $n\rightarrow\infty.$
Fano's Inequality \cite{CoverThomasBook2006} then implies that
$H(X_{n}|\breve{Z}_{n})\rightarrow0$ as required.
\end{proof}
We have thus shown in the context of Theorem \ref{thm:ZgXasymp} that the limit $\nu(Z|X)\rightarrow0$
characterises a regime of large signal-to-noise for the input
$X$, without the need to impose further conditions on the joint distribution of $X$ and $Z$. In \cite{Prelov2004}, the authors consider the opposite regime of low
signal-to-noise for non-linear channels with additive Gaussian noise.
They obtain an asymptotic expansion of the mutual information, $I(X;Z)$, whose leading term
is a decreasing function of a variable which, in their setting, is
equal to $\nu(Z|X)$. In Theorem \ref{thm:XgZasymp} below, we consider
the case where the conditioning is on the response $Z$ instead of on the input $X$ (see Section \ref{depreg} for further discussion of regression on the response variable).
\begin{theorem}
\label{thm:XgZasymp}Let $(X_{n},Z_{n})$ be a sequence of pairs of real-valued
random variables, with the support of $X_{n}$ given by a finite set $\mathcal{X}_{n}$
($|\mathcal{X}_{n}|\geq2$ and bounded above by a constant $\forall n$).
Define $\breve{X}_{n}\triangleq X_{n}\mathbb{V}[X_{n}]^{-\frac{1}{2}}$,
with support \textup{$\mathcal{\breve{X}}_{n}$, }and let $\epsilon_{n}^{\ast}$
be 1/2 of the minimum distance between any two points in $\breve{\mathcal{X}}_{n}$.
Suppose that $\epsilon^{\ast}\triangleq\inf\{\epsilon_{n}^{\ast}\}>0.$
If $\nu(X_{n}|Z_{n})\rightarrow0$ as $n\rightarrow\infty$, then
$H(X_{n}|Z_{n})\rightarrow0$ and $H(X_{n})-I(X_{n};Z_{n})\rightarrow0$. \end{theorem}
\begin{proof}
Given in the Appendix, using an argument similar to the proof of Theorem \ref{thm:ZgXasymp}.
\end{proof}

\noindent Again, the mutual
information $I(X;Z)$ converges to the entropy of $X$, this time as $\nu(X|Z)\rightarrow0$.
(We note in passing that, if both $X$ and $Z$ have finite support, a corollary
of Theorem \ref{thm:XgZasymp} is that: either $H(X_{n})-H(Z_{n})\rightarrow0,$
or $\nu (X_{n}|Z_{n})$ and $\nu (Z_{n}|X_{n})$
do not simultaneously converge to zero).

Theorems \ref{thm:ZgXasymp} and \ref{thm:XgZasymp} establish 
connections between regression-based measures of dependence $\nu ^{-1}$ and mutual information,
without imposing strong assumptions about the joint distribution of
$(X,Z).$ We conjecture that similar theorems will hold for random variables
$X$ and $Z$ having a joint density with respect to Lebesgue measure. These theorems indicate the possibility, explored below, of lower bounding the mutual information using $\nu ^{-1}$.
A consequence of our Theorem \ref{thm:ZgXlwrbnd}
is that, for the general multivariate, absolutely continuous case,
$I(X;Z)\rightarrow\infty$ when $\nu(Z|X)\rightarrow 0$
(and $\mathbb{E}[Z|X=x]$
is an invertible mapping), and $\nu$ is evaluated under the appropriate
marginal distribution for the input $X$.

\section{Lower bounds on mutual information and capacity}

Our aim in this and the subsequent section is to establish lower bounds on mutual information, $I(X;Z)$,
for an absolutely continuous random vector $(X,Z)$ with finite dimension
$d\geq2$. These bounds will hold under certain marginal distributions for the input $X$, and also provide lower bounds on the information capacity. We first show that the following Lemma holds when the marginal
distribution of $X$ is normal. The Lemma relates the mutual information of $X$ and $Z$ to their variances and covariance.

\begin{lem}
\label{lem:lbndCorr}Let $(X,Z)$ be a random vector in $\mathbb{R^{\mathrm{\mathit{d}}}}$,
$d\geq2$, with joint density $f(x,z)$ with respect to Lebesgue
measure. Suppose the density of $X$, $f(x),$ is multivariate normal and that
$(X,Z)$ has finite variance matrix under $f$. Then 
\begin{eqnarray}
I_{f}(X;Z) &\geq& \mathrm{\mathrm{\log}\left\{ \frac{\mathrm{\det}(\Sigma_{\mathit{XX}})}{det(\Sigma_{\mathit{XX}}-\Sigma_{\mathit{XZ}}\Sigma\mathit{_{ZZ}^{-\mathrm{1}}}\Sigma_{\mathit{ZX}})}\right\}_{\mathit{f}} ^{1/2}}  \label{eq:lbndCorr} \\
&=&\mathrm{\log}\left\{ \frac{\mathrm{\det}(\Sigma_{\mathit{ZZ}})}{\mathrm{det}(\Sigma_{\mathit{ZZ}}-\Sigma_{\mathit{ZX}}\Sigma\mathit{_{XX}^{-\mathrm{1}}}\Sigma_{\mathit{XZ}})}\right\}_{\mathit{f}} ^{1/2} \notag
\end{eqnarray}
where, for example, $\Sigma_{ZX}=\mathrm{Cov}(X,Z)=\mathbb{E}\{(X-\mathbb{E}[X])(Z-\mathbb{E}[Z])^{\mathrm{T}}\},$
and subscript $\mathit{f}$ indicates that the mutual information and the covariance matrices are those under
the joint density $f(x,z).$\end{lem}
\begin{proof}
Given in the Appendix.
\end{proof}
When $d=2$ (and, again, $f(x)$ is normal), Lemma \ref{lem:lbndCorr}
simplifies to give the lower bound \[I(X;Z)\geq\mathrm{log\{[1-Corr^{2}(\mathit{X,Z})]^{-1/2}\}},\]
where $\mathrm{Corr}$ denotes the Pearson correlation. This inequality shows how to relate perhaps the best known measure of association to mutual information. An outline
proof is given for the $d=2$ case in \cite{Mitra2001}.
The multivariate,
$d>2$, case has not to our knowledge appeared previously, although
it is a reasonably straightforward generalisation. We provide a complete
proof for $d\geq2$ in the Appendix.
\begin{rem}
The lower bounds in Eq. \ref{eq:lbndCorr} are the mutual information
of a multivariate normal density, $g$, with identical covariance
matrix to that of $f,$ the true joint density of $X$ and $Z$. We want to obtain a lower bound for $I_{\mathit{f}}(X;Z)$
in terms of $\nu_{f}(Z|X).$ It might appear that one way to do so would
be to adopt a similar strategy, attempting to show that $I_{\mathit{f}}(X;Z)$
is bounded below by the mutual information of a multivariate normal,
$g$, having identical $\nu(Z|X)$ to $f$ (and with the marginal
density $g(x)=f(x)$). However, this strategy fails. Although $I_{g}(X;Z)$
still depends only on $\nu(Z|X),$ the equality $\mathbb{E_{\mathit{f}}}\{\mathrm{log}[g(X,Z)/f(X)g(Z)]\}=\mathbb{E_{\mathit{g}}}\{\mathrm{log}[g(X,Z)/f(X)g(Z)]\}$
no longer holds in general (see Eq. \ref{eq:careMS-1} and subsequent argument in the Appendix).
\end{rem}
Our general strategy to obtain lower bounds on mutual information and capacity is as follows. First, we specify a channel with suitable `pseudo-output' and Gaussian `pseudo-input'. These are transformed versions of $Z$ and $X$ respectively (sometimes we transform $X$ alone).  The transformations are given by the conditional mean functions: for example, the pseudo-input may be $m(X) = \mathbb{E}[Z|X]$, as in Theorem \ref{thm:ZgXlwrbnd} below. Second, we apply Lemma \ref{lem:lbndCorr} with the pseudo-input and pseudo-output in place of $X$ and $Z$ there. We then
make use of the following relationship that holds for any random vector
$(U,W)$: 
\begin{eqnarray}
\mathrm{Cov}(U,\mathbb{E}[U|W]) & = & \mathbb{E}\left\{\mathbb{E}\left[(U-\mathbb{E}[U])(\mathbb{E}[U|W]-\mathbb{E}[U])^{\mathrm{T}}|W\right]\right\} \label{eq:CovVE} \\
 & = & \mathbb{V}\{\mathbb{E}[U|W]\}. \notag 
\end{eqnarray}

We now use this strategy to obtain a lower bound for the mutual information, $I_{\mathit{f}}(X;Z)$,
in terms of $\nu_{f}(Z|X).$

\begin{theorem}
\label{thm:ZgXlwrbnd}Let $(X,Z)$ be a random vector in $\mathbb{R^{\mathrm{\mathit{d}}}},d\geq2$.
Consider the conditional density (with respsect to Lebesgue measure), $f(z|x)$,
and suppose that $m(x)=\mathbb{E}[Z|X=x]$ is a one-to-one, continuously
differentiable mapping (whose domain is an open set and a support
of $X$). Let $m(X)$ be normally distributed and denote by $f(x)$ the implied density of $X$. Then 
\begin{equation}
I_{f}(X;Z)\geq\mathrm{\log}\left\{\nu_{f}(Z|X)^{-1/2}\right\},\label{eq:ZgXlwrbound}
\end{equation}
where we assume moments are finite such that $\nu_{f}(Z|X)$ is well defined
and that $\mathbb{E}\mathit{\left\{ \mathbb{V}\left[Z|X\right]\right\} }$
is non-singular under the joint density, $f$. This lower bound is sharp since $I_{f}(X;Z)=\log\left\{\nu_{f}(Z|X)^{-1/2}\right\}$
when $f(x,z)$ is a multivariate normal density. Furthermore, the information capacity $C_{Z|X}$ satisfies
\begin{equation}
C_{Z|X}\geq\mathrm{\log}\left\{\nu_{f}(Z|X)^{-1/2}\right\}, \label{capacitybnd}
\end{equation}
provided $f(x)$ is the density of an allowed input distribution, that is a distribution in $\mathcal{F}$.
\end{theorem}

\begin{proof}
Consider the mechanism $M\rightarrow X\rightarrow Z$, where $X=m^{-1}(M)$
and $X$ is then applied to the channel $f(Z|X).$ Here $M=m(X)$ is the transformed or pseduo-input. Notice that $\mathbb{E}[Z|M]=\mathbb{E}[Z|X]=M$ because $\sigma(X) = \sigma(M)$,
and therefore $\mathrm{Cov}(M,Z)=\mathrm{Cov}(Z,\mathbb{E}[Z|M])=\mathbb{V}\{\mathbb{E}[Z|M]\}=\mathbb{V}[M]$,
by Eq. \ref{eq:CovVE} (and using twice that $\mathbb{E}[Z|M]=M$). Under the conditions of Theorem
\ref{thm:ZgXlwrbnd}, $M$ has a Gaussian distribution, with the distribution
of $X$ that implied by the mapping $X=m^{-1}(M).$ Given the properties
of the one-to-one mapping $m(x)$, we also have that $I_{f}(M;Z)=I_{f}(X;Z)$ (see, for example, \cite{Kraskov2004}).
Applying Lemma \ref{lem:lbndCorr} to the Gaussian, pseudo-input $M$ and the output $Z$ yields 
\begin{eqnarray*}
I_{f}(X;Z)\geq\frac{1}{2}\mathrm{\log}\left\{ \frac{\mathrm{\det}(\mathbb{V}[Z])}{\mathrm{det}(\mathbb{V}[Z]-\mathrm{Cov}(Z,M)\mathbb{V}[M]^{-1}\mathrm{Cov}(Z,M)^{\mathrm{T}})}\right\}_{\mathit{f}} \\
=\frac{1}{2}\mathrm{\log}\left\{ \frac{\mathrm{\det}(\mathbb{V}[Z])}{\mathrm{det}(\mathbb{V}[Z]-\mathbb{V}\{\mathbb{E}[Z|M]\})}\right\}_{\mathit{f}} ,
\end{eqnarray*}
where the second line again uses $\mathrm{Cov}(Z,M)=\mathbb{V}[M]=\mathbb{V}\{\mathbb{E}[Z|M]\}$. The result then follows directly because $\mathbb{V}[Z]=\mathbb{V}\{\mathbb{E}[Z|M]\}+\mathbb{E}\{\mathbb{V}[Z|M]\},$
and because we have that $\mathbb{V}[Z|M]=\mathbb{V}[Z|X]$ since $\sigma(M)=\sigma(X).$
Eq. \ref{capacitybnd} follows from Eq. \ref{eq:ZgXlwrbound} and the definition of capacity as the supremum of mutual information over the collection of allowed input distributions, $\mathcal{F}$.
\end{proof}

Notice that characterising the first and second conditional moments of the mechanism, $\mathbb{E}[Z|X]$ and $\mathbb{V}[Z|X]$, is enough (using the law of total variance) to evaluate the lower bound $\mathrm{\log}\left\{\nu(Z|X)^{-1/2}\right\}$ for a given Gaussian distribution of $m(X)$. Maximising the bound over such distributions then yields the largest lower bound on the information capacity. The approach is applicable when
experimental data have been generated under some other
input distribution, provided the first and second conditional moments are carefully estimated. 

We now discuss the relationship of our lower bounds in Eqs. \ref{eq:ZgXlwrbound} and \ref{capacitybnd} to average mean-square error, rate distortion-type arguments with a Gaussian source.

\begin{cor}
Let $d_{Z}$ be the dimension of $Z$ and suppose that the conditions of Theorem \ref{thm:ZgXlwrbnd} apply. It follows from Eq. \ref{eq:ZgXlwrbound},
scaling by the reciprocal of $d_{Z}$, that

\begin{eqnarray}
d_{Z}^{-1}I_{f}(X;Z)\geq\frac{1}{2}d_{Z}^{-1}\mathrm{\log}\{\nu_{f}(Z|X)^{-1/2}\} & = & \frac{1}{2}\mathrm{\log}\left\{ \frac{\mathrm{[\det}(\mathbb{V}[Z])]^{d_{Z}^{-1}}}{[\mathrm{det}(\mathbb{E}\{\mathbb{V}[Z|X]\})]^{d_{Z}^{-1}}}\right\}_{\mathit{f}} \nonumber \\
\geq &  & \frac{1}{2}\mathrm{\log}\left\{ \frac{\mathrm{[\det}(\mathbb{V}[Z])]^{d_{Z}^{-1}}}{d_{Z}^{-1}\mathrm{tr}(\mathbb{E}\{\mathbb{V}[Z|X]\})}\right\}_{\mathit{f}} ,\label{eq:rdtype}
\end{eqnarray}
 since $\mathbb{E}_{f}\{\mathbb{V}[Z|X]\}$ is positive definite, where $\mathit{f}$ indicates evaluation under the joint density $f(x,z)$. We have used that $\mathrm{0<[\det}(M)]^{1/m}\leq m^{-1}\mathrm{tr}(M)$
for a positive definite, $m\times m$ matrix, $M$ \cite{CoverThomasBook2006}.
\label{rdarg}
\end{cor}
Notice that $d_{Z}^{-1}\mathrm{tr}(\mathbb{E}\{\mathbb{V}[Z|X]\})$
is the average minimum MSEE given $X$, the average being across the scalar components of the vector
$Z$.
Eq. \ref{eq:rdtype} establishes that our lower bound, $\mathrm{\log}\{\nu_{f}(Z|X)^{-1/2}\}$,
is tighter than the lower bound based on the average minimum MSEE, $\frac{1}{2}\mathrm{\log\{\mathrm{\det}(\mathbb{V}[Z])}^{d_{Z}^{-1}}/$ $d_{Z}^{-1}\mathrm{tr}(\mathbb{E}\{\mathbb{V}[Z|X]\})\}_{\mathit{f}}$. The two lower bounds clearly coincide for the bivariate case, $d=2$. The bound based on the average minimum MSEE might, at first sight, appear to have the form that would be obtained by a rate distortion-type argument \cite[Section 4.5.2]{BergerBook} with $Z$ as the Gaussian `source'. However, this is not the case because the bound would in general be evaluated under a joint density for $(X,Z)$ different than $\mathit{f}$. (However, see Lemma \ref{thm:XgZlwrbnd} and the subsequent discussion for the case of regression on the response variable). Notice also that in our setting of Eq. \ref{picture}, we may not be able to adjust the input distribution in order to obtain a Gaussian marginal for $Z$.

Instead, one might treat $M=\mathbb{E}[Z|X]$ as the Gaussian source in a rate distortion-type argument. Consider the case $d=2$. One obtains the result $I_{f}(X;Z)\geq\frac{1}{2}\mathrm{log}(\mathbb{V}\{\mathbb{E}[Z|X]\}/\mathbb{E}\{\mathbb{V}[Z|X]\})_{f}$,
where the numerator is the variance of the source and the denominator
is the expected square-error distortion between the source and its
estimate, here $Z$. The right-hand side of this inequality is strictly
less than our lower bound, $\mathrm{\log}\{\nu_{f}(Z|X)^{-1/2}\}=\frac{1}{2}\mathrm{log}(\mathbb{V}[Z]/\mathbb{E}\{\mathbb{V}[Z|X]\})_{f},$
since $\mathbb{V}[Z]-\mathbb{V}\{\mathbb{E}[Z|X]\}=\mathbb{E}\{\mathbb{V}[Z|X]\}>0$.
An analogous argument applies to the case $d>2$, since $\mathrm{\det}(\mathbb{V}[Z]) > \mathrm{\det}($ $\mathbb{V}\{\mathbb{E}[Z|X]\})$. (For the case where $X$ itself is Gaussian, see Lemma \ref{thm:XgZlwrbnd}).
We conclude that
our lower bounds in Eqs. \ref{eq:ZgXlwrbound} and \ref{capacitybnd} are tighter than lower bounds derived using average mean-square error rate distortion-type arguments with a Gaussian source.

\section{\label{depreg} Regression on the response variable}

We have so far considered lower bounds on mutual information that rely on the error in estimation of the response variable $Z$, using the conditional mean of $Z$ given $X$. In this section we instead consider lower bounds that utilise the regression of $X$ on the response variable, $Z$. 
Using the novel proof strategy adopted for Theorem \ref{thm:ZgXlwrbnd}, we can show that the lower bound $\mathrm{\log}\{\nu_{f}(\tilde{X}|Z)^{-1/2}\}$ improves upon the one based on the Pearson correlation, $\mathrm{log\{[1-Corr_{\mathit{f}}^{2}(\mathit{\tilde{X},Z})]^{-1/2}\}}$, where $\tilde{X}$ is the result of transforming the input to have a Gaussian marginal distribution. We can thus always (weakly) improve upon the lower bound for the bivariate case given by Lemma \ref{lem:lbndCorr}. Intuitively, the improvement arises because $\nu_{f}(\tilde{X}|Z)^{-1}$ captures dependence from non-linearity in the conditional mean, $\mathbb{E}[\tilde{X}|Z]$, whereas the Pearson correlation does not.

\begin{theorem}
\label{thm:CfMS}Let $(X,Z)$ be a random vector in $\mathbb{R^{\mathrm{\mathit{2}}}}$
with joint density $f(x,z)$ with respect to Lebesgue measure. Suppose
there exists a one-to-one mapping $s:X\rightarrow\tilde{X}$ (whose
domain is an open set and a support of $X$) such that $\tilde{X}$
has a Gaussian density. We assume that the mapping $s$ is continuously
differentiable with derivative that is everywhere non-zero, and that
$\nu_{f}(\tilde{X}|Z)^{-1/2}$ exists. Then\textup{ 
\begin{equation}
I_{f}(X;Z)\geq\mathrm{\log}\{\nu_{f}(\tilde{X}|Z)^{-1/2}\}\geq\mathrm{log}\{[1-\mathrm{Corr_{\mathit{f}}^{2}}(\tilde{X},Z)]^{-1/2}\},\label{eq:CfMS}
\end{equation}
}where the third term is the lower bound on $I_{f}(\tilde{X};Z)=I_{f}(X;Z)$
given by Lemma \ref{lem:lbndCorr} in the case $d=2$. We assume
$\mathrm{Corr}_{\mathit{f}}(\mathit{\tilde{X}},Z)$ is well defined and less than $1$. Subscript $f$ indicates evaluation under the joint distribution implied by $f(x,z)$.
\end{theorem}
As we show in the proof below, the lower
bound $\mathrm{\log}\{\nu_{f}(\tilde{X}|Z)^{-1/2}\}$ can be understood
as the result of first transforming $Z$ to the `pseudo-output' $t(Z)=\mathbb{E}[\tilde{X}|Z],$
and then basing the bound on $\mathrm{Corr_{\mathit{f}}^{2}}(\tilde{X},t(Z))$,
using Lemma \ref{lem:lbndCorr}. We then establish
that using another (measurable) transformation,
$t(Z)$, cannot yield a greater bound (given some choice of the Gaussian
variable $\tilde{X}$). This includes, in particular, the lower bound
based on the squared correlation of $\tilde{X}$ and $Z$ itself.
Notice that we cannot construct a lower bound using the maximal correlation
of $(X,Z)$ \cite{Gebelein1941,Renyi1959}, because the implied transformation
of $X$ need not result in the Gaussian distribution needed to apply
Lemma \ref{lem:lbndCorr}. 
\begin{proof}
Given the properties of the one-to-one mapping $s(x)$, we have that
$I_{f}(X;Z)=I_{f}(\tilde{X};Z)$.
Consider the mechanism $\tilde{X}\rightarrow X \rightarrow Z\rightarrow\mathbb{E}[\tilde{X}|Z]$,
in which we first transform $\tilde{X}$ to $X$ and then apply $X$
to the `channel' $f(Z|X).$ The pseudo-output here is $\mathbb{E}[\tilde{X}|Z]$. By the data processing inequality, $I_{f}(\tilde{X};Z)\geq I_{f}(\tilde{X};\mathbb{E}[\tilde{X}|Z]).$
Applying Lemma \ref{lem:lbndCorr} to the Gaussian pseudo-input $\tilde{X}$ and the pseudo-output $\mathbb{E}[\tilde{X}|Z]$ yields 
\[
I_{f}(\tilde{X};\mathbb{E}[\tilde{X}|Z])\geq\frac{1}{2}\mathrm{\mathrm{\log}\{[1-\mathrm{Corr^{2}}(\tilde{X},\mathbb{E}[\tilde{X}|Z])]^{-1/2}}\}.
\]
By Eq. \ref{eq:CovVE}, $\mathrm{Cov}(\tilde{X},\mathbb{E}[\tilde{X}|Z])=\mathbb{V}\{\mathbb{E}[\tilde{X}|Z]\}.$
Hence $\mathrm{Corr^{2}}(\tilde{X},\mathbb{E}[\tilde{X}|Z])=\mathbb{V}\{\mathbb{E}[\tilde{X}|Z]\}$ $/\mathbb{V}[\tilde{X}]$,
and 
\[
I_{f}(\tilde{X};\mathbb{E}[\tilde{X}|Z])\geq\frac{1}{2}\mathrm{\mathrm{\log}\left\{ \mathbb{E\mathit{\left\{ \mathbb{V}\left[\tilde{X}|Z\right]\right\} }}/\mathbb{V}\left\{\mathit{\tilde{X}}\right\}\right\} }=\mathrm{\log}\{\nu_{f}(\tilde{X}|Z)^{-1/2}\},
\]
since $\mathbb{V}[\tilde{X}]=\mathbb{V}\{\mathbb{E}[\tilde{X}|Z]\}+\mathbb{E}\{\mathbb{V}[\tilde{X}|Z]\}$
by the law of total variance. This establishes the first inequality
in Eq. \ref{eq:CfMS} and, importantly, does so in a way that enables
us to establish the second. It is a direct consequence of \cite{Renyi1959}
that $\mathbb{V}\{\mathbb{E}[\tilde{X}|Z]\}/\mathbb{V}[\tilde{X}]$ 
is equal to $\mathrm{sup_{\mathit{t}}\mathrm{Corr}^{2}(\mathit{t\left(Z\right),\tilde{X}})}$,
where the supremum is over all Borel measurable functions $t$ such
that $\mathrm{Corr}(\mathit{t\left(Z\right),\tilde{X}})$ is well
defined. It follows that 
\begin{equation}
1-\nu_{f}(\tilde{X}|Z)=\mathbb{V}\{\mathbb{E}[\tilde{X}|Z]\}/\mathbb{V}[\tilde{X}]\geq\mathrm{Corr}^{2}(\mathit{t(Z),\tilde{X}})\geq\mathrm{Corr}^{2}(\mathit{Z,\tilde{X}}), \label{eq:biv}
\end{equation}
for all $t(\cdot)$, which implies the second inequality in Eq. \ref{eq:CfMS}.
\end{proof}

An analogue of Theorem \ref{thm:ZgXlwrbnd} when the conditioning is on the response variable $Z$ is given by the following Lemma. The proof
is straightforward. A related lower bound is given
without proof in the frequency domain by \cite{Rieke1999}, the bound being on the mutual information
rate in continuous-time.
\begin{lem}
\label{thm:XgZlwrbnd}Let $(X,Z)$ be a random vector in $\mathbb{R^{\mathrm{\mathit{d}}}}$,
$d\geq2,$ with joint density with respect to Lebesgue measure, $f(x,z)$.
Suppose there exists a one-to-one mapping $s:X\rightarrow\tilde{X}$
(whose domain is an open set and a support of $X$) such that $\tilde{X}$
has a Gaussian density. We assume that the mapping $s$ is continuously
differentiable with a Jacobian that is everywhere non-zero, and that
$\nu_{f}(\tilde{X}|Z)^{-1/2}$ exists. Then, scaling by the reciprocal 
of the dimension of $X$,
\begin{equation}
d_{X}^{-1}I_{f}(X;Z)\geq\frac{1}{2}d_{X}^{-1}\mathrm{\log}\{\nu_{f}(\tilde{X}|Z)^{-1/2}\}\geq\frac{1}{2}\mathrm{\log}\left\{ \frac{\mathrm{[\det}(\mathbb{V}[\tilde{X}])]^{d_{X}^{-1}}}{d_{X}^{-1}\mathrm{tr}(\mathbb{E}\{\mathbb{V}[\tilde{X}|Z]\})}\right\}_{\mathit{f}} ,\label{eq:XgZlwrbound}
\end{equation}
and
\begin{equation}
C_{Z|X}\geq I_{f}(X;Z) \geq \mathrm{\log}\left\{\nu_{f}(\tilde{X}|Z)^{-1/2}\right\}, \label{capacitybnd2}
\end{equation}
provided $f(x)$ is the density of an allowed input distribution, that is a distribution in $\mathcal{F}$.
\end{lem}
\begin{proof}
A concise proof of the first inequality in Eq. \ref{eq:XgZlwrbound} uses that $h(\tilde{X}|Z=z)\leq\frac{1}{2}\mathrm{log}[(2\pi e)^{d_{X}}\mathrm{det}(\mathbb{V}[\tilde{X}|Z=z])]$,
which follows from the maximum entropy property of the multivariate
Gaussian distribution for a given covariance matrix. Then $h(\tilde{X}|Z)\leq\frac{1}{2}\mathrm{log}[(2\pi e)^{d_{X}}\mathrm{det}(\mathbb{E}\{\mathbb{V}[\tilde{X}|Z=z])\}]$,
where we have applied Jensen's
inequality, using the concavity of the function $\mathrm{log}\{\mathrm{det}(\Sigma)\}$
for symmetric, non-negative definite matrices, $\Sigma$ \cite{Cover1988}.
The result follows since $I_{f}(X;Z)=I_{f}(\tilde{X};Z)=h(\tilde{X})-h(\tilde{X}|Z)$
and the entropy of the Gaussian $\tilde{X}$ is given by $h(\tilde{X})=\frac{1}{2}\mathrm{log}[(2\pi e)^{d_{X}}\mathrm{det}(\mathbb{V}[\tilde{X}])]$.
The second inequality in Eq. \ref{eq:XgZlwrbound} follows because we
take the covariance matrix $E\{\mathbb{V}[\tilde{X}|Z]\}$ to be positive
definite and $\mathrm{0<[\det}(M)]^{1/m}\leq m^{-1}\mathrm{tr}(M)$
for a positive definite, $m\times m$ matrix, $M$. 
\end{proof}
\noindent The existence of the mapping $s(x)$ in Lemma \ref{thm:XgZlwrbnd} is
not unduly restrictive. For example, when $d=2$ and the input $X$ has a strictly
increasing distribution function taking values in $(0,1)$, then $F_{X}(X)$
is uniformly distributed and can be invertibly transformed to a Gaussian
random variable. Similar comments apply when $d>2$, using the multivariate transformation
of \cite{Rosenblatt1952} to independent uniform r.v.'s on $(0,1)$.

Rate distortion-type arguments using average mean-square error distortion and $\tilde{X}$ as the Gaussian source cannot establish Eq. \ref{eq:XgZlwrbound} for $d_X > 1$. Such arguments \cite[Section 4.5.2]{BergerBook} show only that $\frac{1}{2}\mathrm{\log\{\mathrm{\det}(\mathbb{V}[\tilde{X}])}^{d_{X}^{-1}}/d_{X}^{-1}$ $\mathrm{tr}(\mathbb{E}\{\mathbb{V}[\tilde{X}|Z]\})\}_{f}$
is a lower bound for $d_{X}^{-1}I_{f}(X;Z)$ under the conditions of
Lemma \ref{thm:XgZlwrbnd}. As we have shown, our lower bound in Eq. \ref{eq:XgZlwrbound}, $\mathrm{\log}\{\nu_{f}(\tilde{X}|Z)^{-1/2}\}$, is tighter than this one derived using average mean-square error rate distortion arguments with Gaussian source. 

\section{Applications}
The lower bounds on mutual information and capacity derived in previous sections will prove useful in at least two types of application: analysing the dependence between input and response vectors using empirical data; and analysing the information capacity of signaling and communication mechanisms for which physical models are available. An illustration of the first type of application using simulated data and further discussion are given immediately below. An existing example of the second type is given in \cite{Mitra2001} which examines the  information capacity of optical fiber communication by employing a lower bound based on the Pearson correlation. Indeed, physical models of a communication mechanism can often be solved for their moments when distributional results are not feasible. For example, models of biomolecular signaling mechanisms are stochastic kinetic models of biochemical reaction networks \cite{BowsherAOS} that can be solved approximately using system-size expansions of the master equation \cite{GrimaEMREJCPpaper2010}. Such expansions can be used to provide fast, computational evaluation of  $\mathbb{E}[Z|X]$ and $\mathbb{V}[Z|X]$ for many (rate) parameter vectors describing the network \cite{iNApaper}. Theorem \ref{thm:ZgXlwrbnd}, Eq. \ref{capacitybnd} can then be used to approximate the capacity of the signaling mechanism and explore its parameter sensitivity. 
\subsection{Lower bound estimation\label{MCarlo}}
Estimation of mutual information rapidly becomes problematic as the dimension, $d$, of $(X,Z)$ grows. A distinct advantage of the lower bounds in Eqs. \ref{eq:ZgXlwrbound} and \ref{eq:XgZlwrbound} based on $\mathrm{log}(\nu^{-1/2})$
is that they are amenable to inference by using nonparametric
regression to estimate the relevant conditional mean. Nonparametric regression and covariance matrix estimation methods for higher dimensions \cite{Friedman1991,Bickel2008} should break down more slowly than mutual information estimation methods as $d$ grows. This is valuable for applications, including
those in systems biology where multiple inputs and outputs
often need to be considered. Estimation of our lower bounds should therefore prove useful for analysing the dependence between input and response in higher dimensions.

The following simulation study demonstrates that use of the lower bounds can substantially improve inference about mutual information when the sample size becomes limited for a given value of $d$. Here we use $d=2$. Inferential procedures for the $d>2$ setting lie beyond the scope of the present paper and will be explored in future work.  The $k$-nearest neighbour point estimator \cite{Kraskov2004}, $\hat{I}_{knn}$, is widely regarded as the leading method for estimation of mutual information using continuously distributed data.  We employ a \emph{composite estimator} defined as the maximum of $\hat{I}_{knn}$ and the lower limit of our bootstrap confidence interval for $\mathrm{log}(\nu^{-1/2})$.
This composite
estimator makes use of our lower bounds to correct erroneous point estimates. We find that the lower bounds are able to provide substantial improvements to the downward bias and root mean square error
(rmse) we report for the nearest-neighbour estimator. 

We assume that we are given data $\left\{ (X_{i},Z_{i});i=1,...,N\right\} $
for independent and identically distributed units and that the distribution of the input $X$ is known (see the discussion following Eq. \ref{picture}). We obtain confidence
intervals, for example, for the lower bound in Eq. \ref{eq:CfMS}
based on $\nu(\tilde{X}|Z)$ as follows. (1)
Obtain fitted values, $\hat{X}_{i}$, for the transformed, Gaussian
input $\widetilde{X}$ by nonparametric estimation of $\mathbb{E}[\tilde{X}|Z]$
using a smoothing spline; (2) Obtain the estimate $1-\hat{\nu}_{{\scriptscriptstyle \tilde{X}|Z}}$
as the ratio of the sample variance of $\hat{X}_{i}$ to the known variance of $\tilde{X}_{i}$ (see Eq. \ref{eq:biv}); (3)
Obtain bias-corrected, accelerated ($\mathrm{BC_{a}}$) bootstrap
confidence intervals \cite{Efron1994} using the estimator $\mathrm{\log}(\hat{\nu}_{{\scriptscriptstyle \tilde{X}|Z}}^{-1/2})$.
Details of the proposed procedure are given in the Appendix. 

\begin{figure}[t!]
\begin{centering}
\includegraphics[width=1\textwidth]{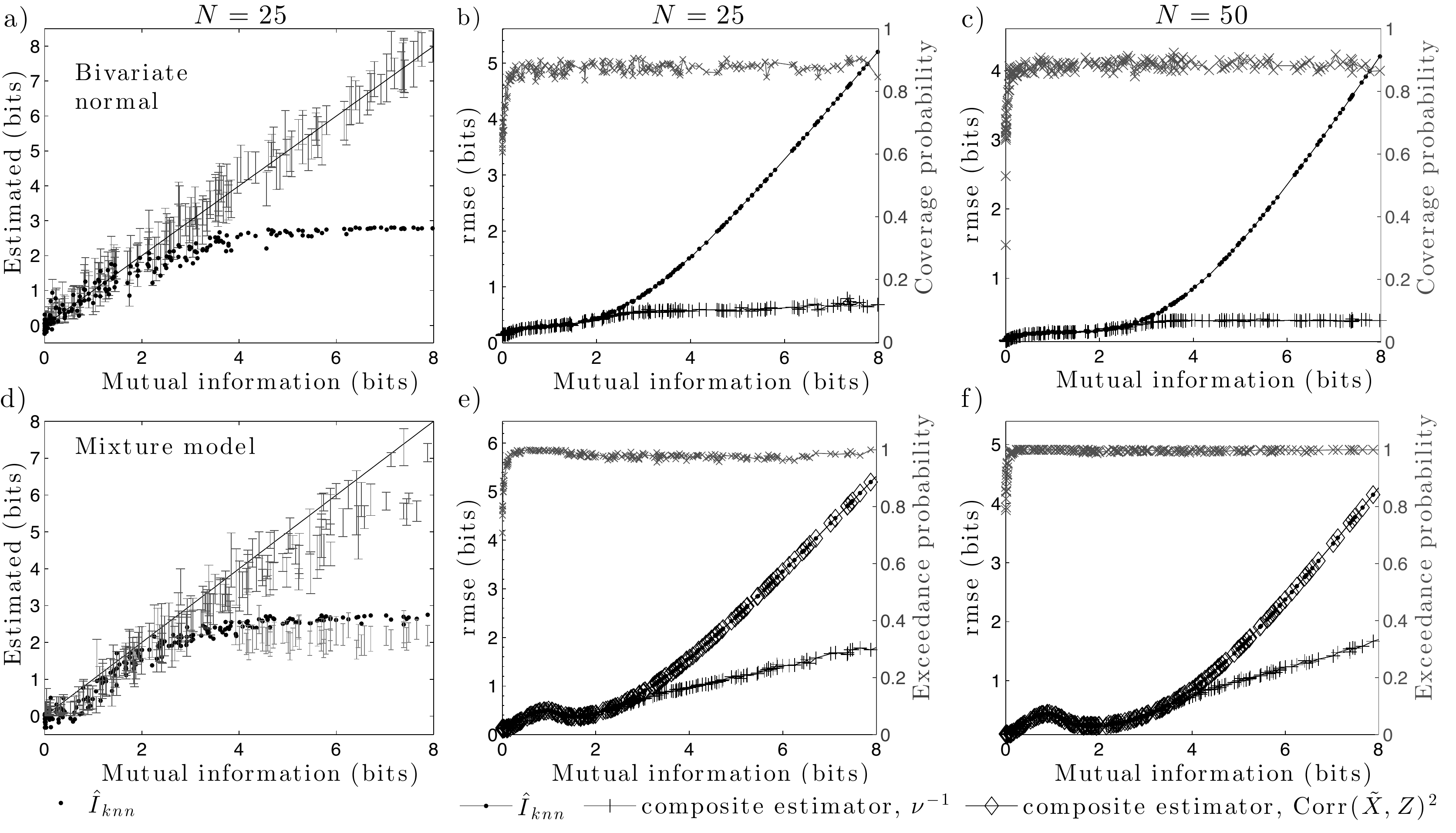}

\par\end{centering}

\caption{ \textbf{Bootstrap confidence intervals of the lower bounds $\mathrm{log}(\nu^{-1/2})$
can substantially improve inference about mutual information through use of a composite estimator. } {{Simulation
results are shown for the bivariate normal (a--c), and for
a mixture model incorporating the same linear regression
model but with a mixture of normals distribution for $X$ (d--f). $N$ is the number of i.i.d. observations.
a) and d): for $N=25$, the $\mathrm{BC_{a}}$, nominally
$90\%$ confidence intervals for our lower bounds in Eqs. \ref{eq:ZgXlwrbound}
and \ref{eq:CfMS} resp., together with the $k$-nearest
neighbour estimates, $\hat{I}_{knn}$, with $k=3$.
(Analogous confidence intervals for the lower bounds based on $\mathrm{Corr}(\tilde{{X}},Z)^{2}$
are shown in d) for the mixture model when $I(X;Z)>4$ bits).
Parameter vectors for each model were sampled independently from their
parameter spaces. A single data set is generated for each
parameter vector in a) and d). Coverage probability
(grey crosses, b) and c)) gives frequency with
which the $\mathrm{BC_{a}}$ interval covers the true $I(X;Z)$; Exceedance
probability (grey crosses, e) and f)) gives frequency
with which $I(X;Z)$ exceeds the lower limit of the $\mathrm{BC_{a}}$
interval (nominally $>0.9$). Root mean square errors (rmse) are plotted
for $\hat{I}_{knn}$ (filled circles), and the composite
estimators (see text) based on $\nu_{{\scriptscriptstyle {Z|X}}}^{-1}$
or $\nu_{{\scriptscriptstyle \tilde{X}|Z}}^{-1}$ (black crosses)
and $\mathrm{Corr}(\tilde{X},Z)^{2}$ (diamonds). Results based
on $500$ Monte Carlo replications.}} }

\label{fig1} \vspace{-0.25in}
 
\end{figure}

Figures 1 and 2 present simulation results for a range of true
values of the mutual information and for two types of data generation
mechanism: a bivariate normal distribution and a mixture model. In
both, $Z=\alpha+\beta X+\varepsilon$, with $\varepsilon$ normally
distributed conditional on $X$ (with constant variance independent
of $X$). In the first, $X$ has a marginal normal distribution (under the data generating density, $f$),
hence $(X,Z)$ has a bivariate normal distribution, $X=\tilde{X}$, and the bounds $\mathrm{log}(\nu^{-1/2})$ in Eqs. \ref{eq:ZgXlwrbound}  and \ref{eq:CfMS}
hold with equality. In the second, $X$
is specified to be an equally-weighted mixture of 2 normals, and we obtain the pseudo-input $\tilde{X}$
by first transforming to uniformity using the probability integral transform
and then transforming to normality. We adopt the second specification  because $\mathbb{E}[\tilde{X}|Z]$
becomes non-linear (and sigmoidal), but the true value of $I(X;Z)$ is
still known with precision through the use of a Monte Carlo average for $h(Z)$ (see Appendix). Details of the parameterisations of the models used are also given in the Appendix.


\begin{figure}[tbh]
\includegraphics[width=1\textwidth]{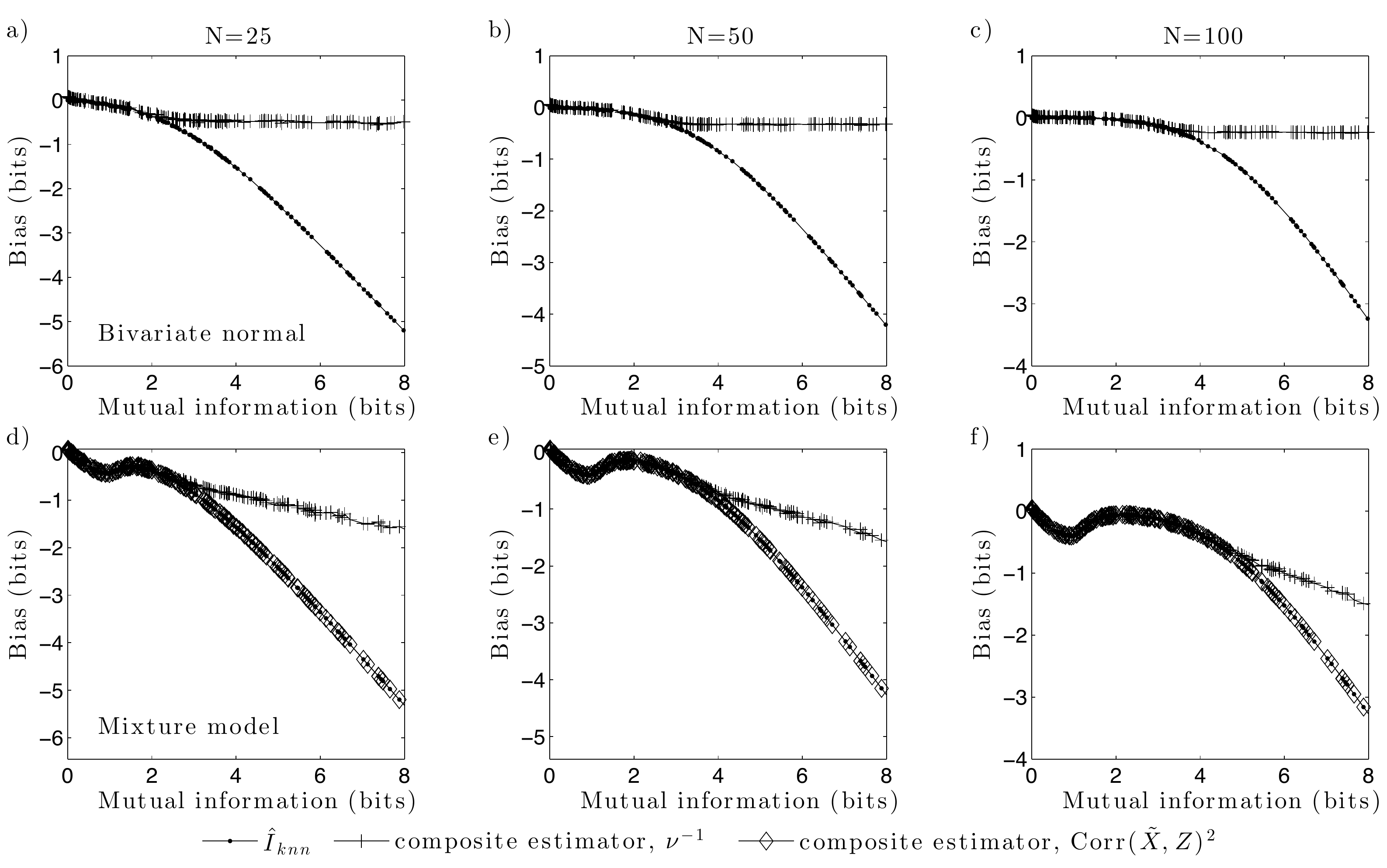}

\caption{\textbf{The lower bounds $\mathrm{log}(\nu^{-1/2})$
can substantially reduce bias through use of a composite estimator}. The lower bounds based on the Peason correlation do not reduce estimation bias for the mixture model with non-linear conditional mean. Biases are plotted for the $k$-nearest neighbour estimates, $\hat{I}_{knn}$,
with $k=3$ (filled circles), and for the composite estimators (see
text) based on $\nu_{{\scriptscriptstyle Z|X}}^{-1}$ or $\nu_{{\scriptscriptstyle \tilde{X}|Z}}^{-1}$
(bivariate normal and mixture models respectively; crosses) and $\mathrm{Corr}(X,Z)^{2}$
(diamonds). Parameter vectors for each model were sampled independently
from their parameter spaces. Results are shown for the
bivariate normal (a--c), and for the mixture model (d--f). $N$ is the number of \emph{i.i.d.} observations in each
data set. All results based on $500$ Monte Carlo replications. }
\end{figure}

Panels a) and d) of Fig. 1 show our $\mathrm{BC_{a}}$ confidence
intervals for the lower bounds based on Eq. \ref{eq:ZgXlwrbound}
and \ref{eq:CfMS} respectively, together with the point estimates
$\hat{I}_{knn}$, for independently generated data sets corresponding
to different true values of $I(X;Z)$ and for sample size $N=25$. In \cite{Kraskov2004},
the authors recommend in practice to use values of $k$ between 2 and 4. We therefore
calculate $\hat{I}_{knn}$ using $k=3$ nearest neighbours ($\hat{I}_{knn}=I^{(2)}(X,Z;k=3)$ in \cite{Kraskov2004}). The poor
performance of $\hat{I}_{knn}$ with this sample size is evident for
both models, particularly for mutual information in excess of 3 bits,
where substantial, growing bias and rmse are evident (see Fig. 2 for  plots of bias).
Higher values of $k$ result in worse bias and rmse of $\hat{I}_{knn}$ (not shown). The remaining panels of Fig. 1 depict,
for sample sizes $N=25$ and $N=50$, various properties under repeated
sampling: the frequency with which the $\mathrm{BC_{a}}$ interval
for the lower bound covers {[}b) and c){]} and has a lower limit exceeded
by {[}e) and f){]} the true mutual information; the rmse of $\hat{I}_{knn}$;
and the rmse of our composite estimator, given by the maximum of $\hat{I}_{knn}$
and the lower limit of the $\mathrm{BC_{a}}$ interval. For both sample sizes, the nonparametric confidence
intervals perform well under repeated sampling and provide substantial
reductions in bias and rmse when comparing $\hat{I}_{knn}$ to the composite
estimator (see also Fig. 2). Finally, in the mixture model where $\mathbb{E}[\tilde{X}|Z]$
is non-linear in $Z$, the lower bound based on $\mathrm{Corr}(\tilde{X},Z)^{2}$
performs considerably worse than that based on ${\nu}_{{\scriptscriptstyle \tilde{X}|Z}}$, as shown in panels d) to f) of Figs. 1 and 2. The corresponding $\mathrm{BC_{a}}$ intervals lie well below those based on ${\nu}_{{\scriptscriptstyle \tilde{X}|Z}}$ and have lower limits below $\hat{I}_{knn}$ in all cases shown in panel d). The associated composite estimator consequently fails to reduce either the bias or the rmse of estimation.


\section{Appendix}
\subsection{Additional proofs}
\begin{proof}
(Lemma \ref{lem:lbndCorr}) Let $g(x,z)$ be the multivariate Gaussian
density with the same unconditional first and second moments as $f(x,z),$
and with marginal Gaussian density $g(x)=f(x).$ Thus, $\mathrm{\mathbb{V}_{\mathit{g}}}[\left(X,Z\right)]=\mathrm{\mathbb{V}_{\mathit{f}}}[\left(X,Z\right)].$
We use subscripts to identify the relevant joint density throughout.
Notice that 
\[
I_{\mathit{f}}(X;Z)=\mathbb{E_{\mathit{f}}}\left\{ \mathrm{log}\frac{g(X,Z)}{f(X)g(Z)}\right\} -\mathbb{E_{\mathit{f}}}\left\{ \mathrm{log}\frac{g(X,Z)f(Z)}{f(X,Z)g(Z)}\right\} \geq\mathbb{E_{\mathit{f}}}\left\{ \mathrm{log}\frac{g(X,Z)}{f(X)g(Z)}\right\} ,
\]
where the second expectation of the equality is seen to be non-positive
by applying Jensen's inequality and then integrating first with respect
to $x$. Furthermore, 
\begin{equation}
\mathbb{E_{\mathit{f}}}\left\{ \mathrm{log}\frac{g(X,Z)}{f(X)g(Z)}\right\} =\mathbb{E_{\mathit{g}}}\left\{ \mathrm{log}\frac{g(X,Z)}{f(X)g(Z)}\right\} =I_{\mathit{g}}(X;Z),\label{eq:careMS-1}
\end{equation}
because $\mathbb{E}_{f}[\mathrm{log\{}g(\cdot)\}]=\mathbb{E}_{g}[\mathrm{log}\{g(\cdot)\}].$
For example, $\mathbb{E}_{f}[\mathrm{log}\{g(X,Z)\}]=\mathbb{E}_{g}[\mathrm{log}\{g(X,Z)\}]$
because 
\begin{eqnarray*}
\mathbb{E}_{f}\left\{ \left(\begin{array}{c}
X-E[X]\\
Z-E[Z]
\end{array}\right)^{\mathrm{T}}\mathrm{\mathbb{V}_{\mathit{g}}}[\left(X,Z\right)]^{-1}\left(\begin{array}{c}
X-E[X]\\
Z-E[Z]
\end{array}\right)\right\}  & =\\
\mathrm{tr}\left\{ \mathrm{\mathbb{V}_{\mathit{g}}}[\left(X,Z\right)]^{-1}\mathbb{E}_{f}\left[\left(\begin{array}{c}
X-E[X]\\
Z-E[Z]
\end{array}\right)\left(\begin{array}{c}
X-E[X]\\ 
Z-E[Z]
\end{array}\right)^{\mathrm{T}}\right]\right\}  & =d,
\end{eqnarray*}
since $\mathrm{\mathbb{V}_{\mathit{g}}}[\left(X,Z\right)]=\mathrm{\mathbb{V}_{\mathit{f}}}[\left(X,Z\right)].$
Evaluating $I_{\mathit{g}}(X;Z)$=$h_{g}(X)+h_{g}(Y)-h_{g}(X,Y)$
is straightforward since the marginal and joint densities under $g$
are all Gaussian. We find 
\[
I_{\mathit{g}}(X;Z)=\frac{1}{2}\mathrm{log\left\{ \frac{det(\mathbb{V_{\mathit{g}}}\left[X\right])det(\mathbb{V_{\mathit{g}}}\left[\mathit{Z}\right])}{det(\mathbb{V_{\mathit{g}}}\left[(X,Z)\right])}\right\} }=\frac{1}{2}\mathrm{log\left\{ \frac{det(\mathbb{V_{\mathit{f}}}\left[X\right])det(\mathbb{V_{\mathit{f}}}\left[\mathit{Z}\right])}{det(\mathbb{V_{\mathit{f}}}\left[(X,Z)\right])}\right\} },
\]
since $g$ and $f$ have identical second moments by construction.
The stated results are then obtained by partitioning of the matrix
$\mathbb{V}_{f}\left[(X,Z)\right].$
\end{proof}

\begin{proof}
(Theorem \ref{thm:XgZasymp}). We have that $\nu(X_{n}|Z_{n})=\nu(\breve{X}_{n}|Z_{n})\rightarrow0$.
Since $\nu(\breve{X}_{n}|Z_{n})=\mathbb{E}\{\mathbb{V}[\breve{X}_{n}|Z_{n}]\}$
$=\mathbb{E}\{(\mathit{\breve{X}_{n}-\mathbb{E}[\breve{X}_{n}|Z_{n}]})^{2}\}$,
it follows that $\breve{X}_{n}-\mathbb{E}[\breve{X}_{n}|Z_{n}]$ converges
to zero in mean square (in $L^{2}$) and therefore $\breve{X}_{n}-\mathbb{E}[\breve{X}_{n}|Z_{n}]\rightarrow^{pr}0$.
Consider estimating $\breve{X}_{n}$ based on observation of $Z_{n}$
as follows: the estimator $\hat{X}_{n}(Z_{n})$ is equal to a point
in the support of $\breve{X}_{n}$ which minimises the Euclidean distance
from $\mathbb{E}[\breve{X}_{n}|Z_{n}]$. Let $x\in\mathbb{R},\breve{x}_{n}\mathcal{\in\breve{X}}_{n}$
and notice that if $|\breve{x}_{n}-x|<\epsilon^{*},$ then $\mathit{|}\breve{x}_{n}-x|<\epsilon_{n}^{*}<\mathit{|}\breve{x}_{n}'-x|$,
that is $x$ is closer to $\breve{x}_{n}$ than to any other point
$\breve{x}_{n}'$ in $\mathcal{\breve{X}}_{n}$. Therefore, if $|\mathit{\breve{X}_{n}-}E[\breve{X}_{n}|Z_{n}]|<\epsilon^{*},$
$\mathbb{E}[\breve{X}_{n}|Z_{n}]$ is closer to $\breve{X}_{n}$ than
to any other point in the support, the estimator $\hat{X}_{n}(Z_{n})$
is uniquely defined, and that estimator recovers $\breve{X}_{n}$
without error $\left(\hat{X}_{n}(Z_{n})=\breve{X}_{n}\right).$ Thus,
the probability of estimation error, $p_{\mathrm{error}},$ satisfies
\[
p\mathrm{_{error}}=P\left\{ \mathit{\hat{X}_{n}\left(Z_{n}\right)\neq\breve{X}_{n}}\right\} \leq\mathbb{P}\left\{ |\mathit{\breve{X}_{n}-}E[\breve{X}_{n}|Z_{n}]|\geq\epsilon^{*}\right\} .
\]
Since $\breve{X}_{n}-\mathbb{E}[\breve{X}_{n}|Z_{n}]\rightarrow^{pr}0$,
$p_{\mathrm{error}}$ must therefore tend to zero as $n\rightarrow\infty.$
Fano's Inequality gives 
\[
H(p\mathrm{_{error}})+p_{\mathrm{error}}\mathrm{log}|\mathcal{X}_{\mathit{n}}|\geq H(\breve{X}_{n}|Z_{n})=H(X_{n}|Z_{n}),
\]
since $|\mathcal{\breve{X}}_{\mathit{n}}|=|\mathcal{X}_{\mathit{n}}|$
and the rescaling does not change the conditional entropy. Therefore
$H(X_{n}|Z_{n})\rightarrow0$ as $n\rightarrow\infty.$
\end{proof}

\subsection*{Models, parametrisations and algorithms used in the simulation study of Section \ref{MCarlo}}

Figures 1 and 2 present simulation results for two data generation
mechanisms. In both, $Z=\beta X+\varepsilon$ with $\varepsilon\sim N(0,\sigma_{\varepsilon}^{2})$
and $\varepsilon$ independent of $X.$ The two models, together with
the schemes used to generate parameter vectors for the results shown
in Figures 1 and 2, are as follows:
\begin{enumerate}
\item \textit{Bivariate normal }\emph{model}: $X\sim N(0,\sigma_{{\scriptscriptstyle{X}}}^{2}).$
Model parameters were sampled as follows: i) $\beta$ uniformly distributed on $(1,10)$; ii) $\sigma_{\varepsilon}^{2}=10^{\theta_{1}}$
with $\theta_{1}$ uniformly distributed on $(-2,2)$; and iii) $\sigma_{{\scriptscriptstyle{X}}}^{2}=10^{\theta_{2}}$
with $\theta_{2}$ uniformly distributed on $(-2,2)$.
\item \textit{Mixture model}: $X$ is an equally-weighted mixture of 2 normal
distributions, that is $f_{{\scriptscriptstyle{X}}}(x)=\frac{1}{2}N(\mu_{1},\sigma_{1}^{2})+\frac{1}{2}N(\mu_{2},\sigma_{2}^{2})$. Model parameters were sampled as
follows: i) $\beta=10^{\theta_{1}}$, with $\theta_{1}$ uniformly
distributed on $(-1,1)$; ii) $\sigma_{\varepsilon}^{2}=10^{\theta_{2}}$
with $\theta_{2}$ uniformly distributed on $(-2.5,2.5)$. We set
$\mu_{1}=-\mu_{2}=5$ and $\sigma_{1}^{2}=\sigma_{2}^{2}=25/4.$
\end{enumerate}

The mixture model allows precise evaluation of $I(X;Z)=h(Z)-h(Z|X)$
via Monte Carlo sampling. We have $h(Z|X)=\frac{1}{2}\mathrm{log(2\pi e\sigma_{\varepsilon}^{2})}$.
Note that the marginal density $f(Z)$ is also an equally-weighted
mixture of 2 normals which we can express in closed form. Hence, we
can also estimate $h(Z)=-\mathbb{E}[\mathrm{log}f(Z)]$ as the Monte
Carlo average of $\mathrm{log}f(Z_{m})$ where $Z_{m}$ ($m=1,...,M$)
is a draw from the mixture model. For our numerical calculations we
set $M=10^{5},$ and monitored convergence of the Monte Carlo average.

Computations were implemented in R (version 2.12.2). The non-parametric
estimation of $\mathbb{E}[\tilde{X}|Z]$ was performed using the `smooth.spline'
function (an implementation of smoothing splines \cite{Green1994}) with
the number of knots set to $10$; the smoothing parameter was
chosen using cross-validation on the original dataset; all other parameters
were set to their default values. $90\%$ $\mathrm{BC_{a}}$ confidence
intervals were calculated from $B=2000$ bootstrap replications (using
the `boot' package). For the $k$-nearest neighbour estimation of mutual information \cite{Kraskov2004}
we used the authors' `MIxnyn' function within their MILCA suite
(available at \url{http://www.klab.caltech.edu/~kraskov/MILCA/}).

\bibliographystyle{nar}
\bibliography{BiometrikaMendeleyCollection}

\end{document}